\documentclass[journal,12pt,draftclsnofoot,onecolumn]{IEEEtran}

\usepackage{graphicx}
\usepackage[center]{caption}
\usepackage{epsfig,latexsym,amsmath,amsfonts}
\usepackage{amssymb}
\usepackage{placeins}
\usepackage{subfigure}
\usepackage{verbatim}
\usepackage{cite,url}
\usepackage{epsf,bm}

\newtheorem{theorem}{Theorem}[section]
\newtheorem{lemma}[theorem]{Lemma}

\begin{document}

\title{On the Stability of Random Multiple Access with Feedback Exploitation and Queue Priority}
\author{\large  Karim G. Seddik\\
\normalsize  \begin{tabular}{c} Electronics Engineering Department, American University in Cairo \\ AUC Avenue, New Cairo 11835, Egypt.\\
email: kseddik@aucegypt.edu
\end{tabular} }

 \maketitle
\begin{abstract}
In this paper, we study the stability of two interacting queues under random multiple access in which the queues leverage the feedback information. We derive the stability region under random multiple access where one of the two queues exploits the feedback information and backs off under negative acknowledgement (NACK) and the other, higher priority, queue will access the channel with probability one. We characterize the stability region of this feedback-based random access protocol and prove that this derived stability region encloses the stability region of the conventional random access (RA) scheme that does not exploit the feedback information.
\end{abstract}

\newpage



\section{Introduction}\label{Int}
The stability of interacting queues has been extensively considered in literature. Several works have considered the characterization of the stability region of interacting queues under random access protocols. The stability region is characterized for the case $M=2$ and $M=3$ interacting queues as well as the case of $M>3$ with symmetric arrivals. The stability region for the general case of $M>3$ with asymmetric arrivals is still an open problem and only inner achievable bounds are known.
 
Recently, many papers have considered the problem of interacting queues in different contexts. For example, \cite{KommalIT} considers the problem of interacting queues in a TDMA system where a relay is used to help the source nodes in forwarding their lost packets. In \cite{ephre-info}, the stability of interacting queues under a random access protocol in the context of \textit{Cognitive Radio Network} was derived. In \cite{ephre-isit}, the stability region of two interacting queues under random access protocol where the two queues harvest energy was characterized. Other works can be found in \cite{Fanous1,Fanous3}, where derivations of the stability regions in the context of different cognitive radio networks were considered.

In this paper, we derive the stability region of a two-queue random access (RA) protocol with priorities. The queues will apply the conventional RA protocol but in the case of packet loss due to collision the two queues will exploit the feedback information to provide some level of coordination. We set a priority to one of the two queues as follows. In the case of a negative acknowledgement, the queue with the higher priority will attempt transmission in the following time slot with probability one and the other queue will back off to allow for collision-free transmission of the higher priority queue. Clearly, this will enhance the service rate for the higher priority queue but more interestingly it will also improve the service rate for the other, less priority queue as will be explained later. We derive an expression for the boundary of the stability region and prove that the RA with priority scheme encloses the stability region of the conventional RA scheme. 

To the best of our knowledge, the problem of characterizing the stability region of the random access protocol with feedback leveraging has not been considered before. We will characterize the stable arrival rates region and prove that it contains that of the conventional random multiple access scheme (with no feedback exploitation).

The rest of the paper is organized as follows. The system model is presented in Section \ref{sysmod}. The performance of the proposed scheme is investigated in Section \ref{perfanal}. The paper is concluded in Section \ref{Concl}. We have moved most of the proofs to the appendices to preserve the flow of ideas in the paper.




\section{System Model}\label{sysmod}
The system model is shown in Fig. \ref{fig1}. We consider the case of two interacting packet queues, namely $Q_1$ and $Q_2$. $Q_1$ and $Q_2$ have infinite buffers for storing fixed length packets. The channel is slotted in time and any slot duration equals one packet transmission time. The arrival processes at the two queues, $Q_1$ and $Q_2$, are modeled as Bernoulli arrival processes with means $\lambda_1$ and $\lambda_2$, respectively \cite{ephre-isit}. Under our system model assumptions, the average arrival rates are $\lambda_1$ and $\lambda_2$ packets per time slot, and are bounded as $0<\lambda_i<1$, $i=1,2$\footnote{The maximum service rate in our model is 1 packet/slot, since the slot duration equals one packet transmission time, then the arrival rates must be less than 1 otherwise the system will be unstable \cite{ephre-isit}.}. We can assume that the packets arrive at the start of the time slot.

\begin{figure}
\centering
\includegraphics[scale=.5]{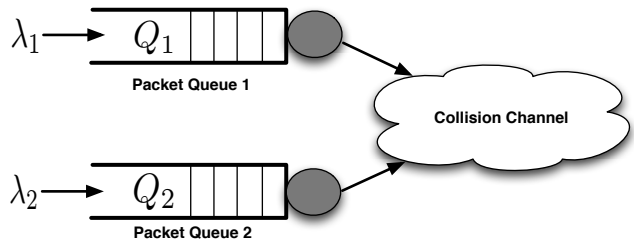}
\caption{The system model.}
\label{fig1}
\end{figure}


The channel is modeled as a collision channel, where packet loss results only in the case of simultaneous transmissions from the two queues. If only one queue attempts transmitting at a given time slot, the packet is considered to be correctly received \cite{Rao88,ephre-isit}. In the random access phase, the first queue accesses the channel with probability $p_1$ whenever it has packets to send and the second queue will access the channel with probability $p_2$ whenever it has packets to send. If at any time slot some queue is empty, it will not attempt any channel access. 

In this paper, we will consider the use of the feedback information that is leveraged at the queues in the case of collision. In the conventional random multiple access system and in the case of collision, the collided packets stay on the head of the queues and retransmissions are attempted employing the same random multiple access scheme. In this paper, we consider a system where the feedback information is leveraged at the queues and a priority is set to the first queue; in the next time slot after collision, queue 2 ($Q_2$) will back off and queue 1 ($Q_1$) will retransmit its collided packet to allow for collision-free transmission of $Q_1$; after that the two queues return to the conventional random multiple access scheme. The priority set to queue 1 can be due to some quality of service (QoS) requirement that is different from the QoS requirement of queue 2. The interesting result is that although the feedback will enhance the service of queue 1 by setting a higher priority to it, the service will be enhanced as well for queue 2 as will be explained later.

\section{The Stability Region for the Feedback-Based Random Access Protocol with Priorities}\label{perfanal}
In this section, we will characterize the stability region for the feedback-based random access scheme. Stability can be loosely defined as having a certain quantity of interest kept bounded. In our case, we are interested in the queue size being bounded. For an irreducible and aperiodic Markov chain with countable number of states, the chain is stable if and only if it is positive recurrent, which implies the existence of its stationary distribution. For a rigorous definition of stability under more general scenarios see
\cite{Rao88} and \cite{Szpan94}.

If the arrival and service processes of a queueing system are strictly stationary, then one can apply Loynes's theorem to check
for stability conditions \cite{Loynes}. This theorem states that if the arrival process and the service process of a queueing system are
strictly stationary, and the average arrival rate is less than the average service rate, then the queue is stable, otherwise it is unstable.

Characterizing the stability region will be a difficult problem due to the interaction of the two queues and due to the fact that the service for one queue will depend on the state of the other queue. We will consider the use of the \textit{Dominant System} concept that was proposed in \cite{Rao88} to characterize the stability region of the conventional RA scheme. We will define two dominant systems tailored to match our feedback-based random access scheme and in each of the two systems we will determine the boundaries of the stability region.   

\subsection{Dominant System 1}
In any dominant system, we define a system that ``stochastically dominates'' our system, that is the queues lengths in the dominant system are always larger than the queues lengths in our system if both, the dominant system and our system, start from the same initial state and have the same arrivals and encounter the same packet collisions.

\begin{figure*}
\centering
\includegraphics[scale=.14]{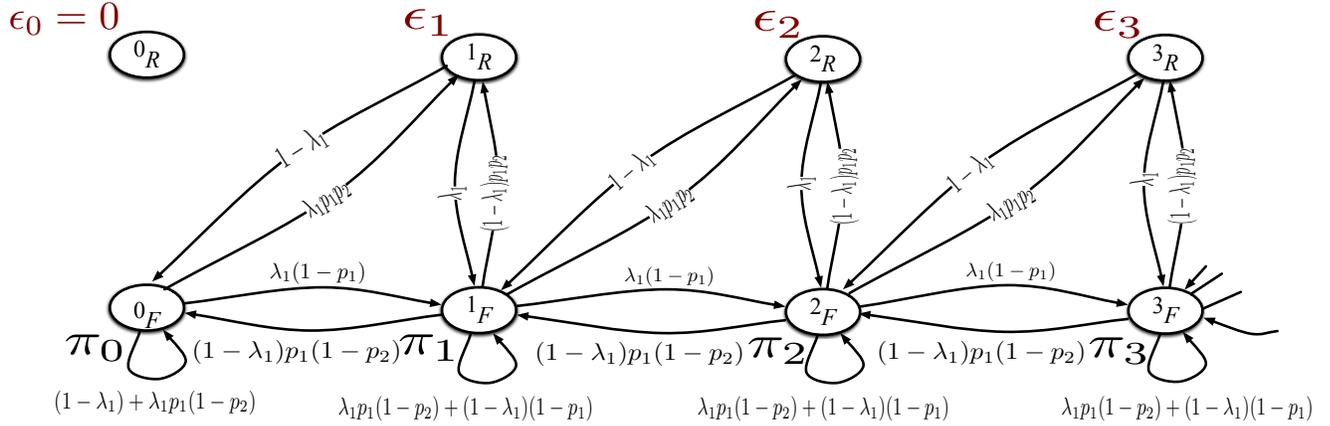}
\caption{Queue 1, $Q_1$, Markov chain model for Dominant System 1.}
\label{fig2}
\end{figure*}

For the first Dominant System, we assume that queue 2 will always have packets to transmit; even if the queue was empty dummy packets will be transmitted from queue 2. Clearly this will set a dominant system to our system since the transmission of dummy packets can only result in more collisions and packet losses. If for a given arrival rate pair ($\lambda_1$, $\lambda_2$) the first dominant system is stable then clearly our system will be stable. Therefore, the stability region of the first dominant system will provide an inner bound for our system stability region.

For queue 1, the Markov chain describing the evolution of the queue is shown in Fig. \ref{fig2}. Note that the Markov chain has two classes of states, namely, $k_F$ and $k_R$ and $k=0,1,2,\cdots$. The subscript $F$ denotes first transmission states and the subscript $R$ denotes retransmission states. Note that in the retransmission states, queue 1 packet will always be delivered since there is no collisions in these states (queue 2 is backing off); in these states, either queue 1 length decreases by 1 if no arrival occurs or the queue length will remain the same if an arrival occurs while being in these retransmission states since the packet on the head of the queue is successfully transmitted with probability 1. 

The stability condition for queue 1 in Dominant System 1 is given in the following lemma, which is proved in Appendix~\ref{app1}.

\begin{lemma}\label{lem1}
The arrival rates for queue 1 and queue 2 in Dominant System 1 must satisfy the following two conditions, respectively,
\begin{equation}
\begin{split}
\lambda_1 &<\frac{p_1}{1+p_1p_2}\\
\lambda_2 &<p_2(1-\lambda_1-\lambda_1p_2)
\end{split}
\end{equation}
for the system to be stable. 
\end{lemma}


\subsection{Dominant System 2}
In the second Dominant System, we assume that queue 1 always has packets to send (dummy packets are sent if the queue decides to transmit while being empty). Again, this will decouple the interaction of the two queues since the service rate of queue 2 will be independent of the state of queue 1.

The Markov chain for the evolution of queue 2 is shown in Fig. \ref{fig3}. Two classes of states are defined in Fig. \ref{fig3} and denoted by the subscripts \textbf{ON} and \textbf{OFF}. The ON states denote the states where queue 2 can access the channel. The OFF states denote the back off states where queue 1 is retransmitting its collided packets. Note that the transitions from the $k_{\text{OFF}}$ state can be either to the $k_{\text{ON}}$ state, if no arrival occurs in the slot, or to the $(k+1)_{\text{ON}}$ state, if one arrival occurs in the slot. The OFF states can never make a transition to a state with a lower number of packets since in the OFF states queue 2 is in the back off mode and no access is attempted. 

The stability condition for queue 2 in Dominant System 2 is given in the following lemma, which is proved in Appendix~\ref{app2} (the analysis in Appendix~\ref{app2} will be based on the theory of homogeneous quasi birth-and-death (QBD) Markov chains \cite{latouche}).

\begin{lemma}\label{lem2}
The arrival rates for queue 1 and queue 2 in Dominant System 2 must satisfy the following two conditions, respectively,
\begin{equation}
\begin{split}
\lambda_1&<\frac{p_1(1-p_1-\lambda_2p_1)}{(1-p_1)}
\\ \lambda_2 &<\frac{p_2(1-p_1)}{1+p_1p_2}
\end{split}
\end{equation}
for the system to be stable. 
\end{lemma}

\begin{figure*}
\centering
\includegraphics[scale=.14]{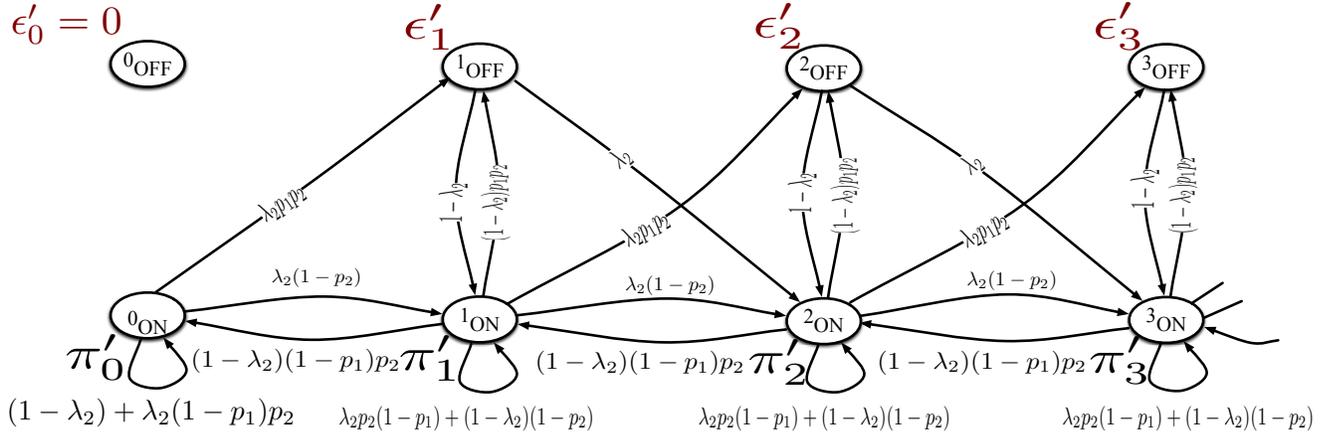}
\caption{Queue 2, $Q_2$, Markov chain model for Dominant System 2.}
\label{fig3}
\end{figure*}

Note that the intersection of the two stability regions described in Lemma \ref{lem1} and Lemma \ref{lem2} for a given access vector $\mathbf{p} =[p_1\;p_2]^T$ (grey area in Fig. \ref{stab}) can be interpreted as follows. Define a new Dominant System (Dominant System 3) in which every queue has always a packet to transmit. In this case, the transmission state of queue 1 can be represented by the two-state Markov chain model shown in Fig~\ref{fig:dom3-1}; note that in this case queue 1 will be either in the ``Transmission'' state denoted by $F$ or in the ``Retransmission'' state denoted by $R$ in Fig.~\ref{fig:dom3-1}. Fig.~\ref{fig:dom3-2} shows the Markov chain model for queue 2. Queue 2 will have two states denoted by ON when queue 1 is in the F state and OFF when queue 1 is the $R$ state (when queue 1 is in the $R$ state queue 2 will be in the back off, OFF state). It is straightforward to show that the steady state distributions for the two Markov chains shown in Fig.~\ref{fig:marin} are given by
\begin{equation}
\begin{split}
\pi_F &= \pi_{\text{ON}}=\frac{1}{1+p_1p_2}\\
\pi_R &= \pi_{\text{OFF}}=\frac{p_1p_2}{1+p_1p_2}.
\end{split}
\end{equation}

\begin{figure}
\centering
\includegraphics[scale=.5]{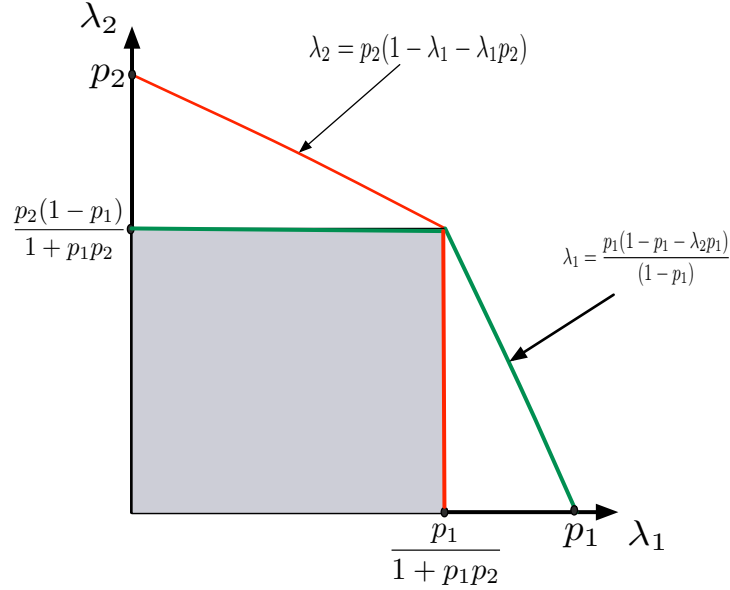}
\caption{The union of the stability regions for the two dominant systems for fixed access probabilities $p_1$ and $p_2$.}
\label{stab}
\end{figure}

The service rate for queue 1 in Dominant System 3, $\mu_{1}''$, is given by
\begin{equation}
\begin{split}
\mu_{1}''=p_1(1-p_2)\pi_F+\pi_R = \frac{p_1}{1+p_1p_2},
\end{split}
\end{equation}
where queue 1 is served with probability $p_1(1-p_2)$ in the $F$ state and with probability 1 in the $R$ state.

The service rate for queue 2 in Dominant System 3, $\mu_{2}''$, is given by
\begin{equation}
\begin{split}
\mu_{2}''=p_2(1-p_1)\pi_{\text{ON}}+0\times \pi_{\text{OFF}} = \frac{p_2(1-p_1)}{1+p_1p_2},
\end{split}
\end{equation}
where queue 2 is served with probability $p_2(1-p_1)$ in the ON state and with probability 0 in the OFF state.

\begin{figure}
 \centering
 \subfigure[The two-state Markov chain model for queue 1 transmission state in Dominant System 3.]{
  \includegraphics[scale=0.18]{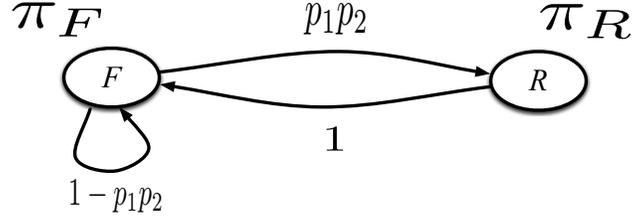}
   \label{fig:dom3-1}
   }
 \subfigure[The two-state Markov chain model for queue 2 transmission state in Dominant System 3.]{
  \includegraphics[scale=0.18]{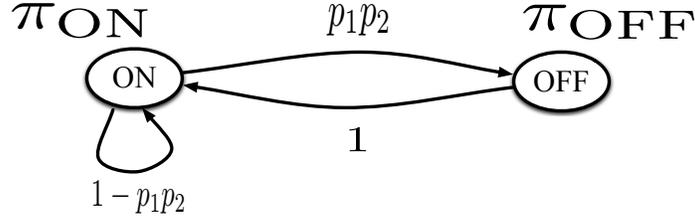}
   \label{fig:dom3-2}
   }
   \caption{Dominant System 3 Markov chain model.}
   \label{fig:marin}
\end{figure}

\subsection{The Stability Region of the Random Access Protocol with Priorities}
In this section, we derive the expression for the stability region of the random access scheme with feedback exploitation where a priority is set to one of the two queues. The following Lemma characterizes the stability region for fixed random access probabilities, $p_1$ and $p_2$, for queue 1 and queue 2, respectively.

\begin{lemma}\label{lem3}
For a fixed random access probability vector $\mathbf{p}=[p_1\; p_2]^T$, the stability region $\mathcal{R}(\mathbf{p})$ of the random access with priorities is the union of the two regions described by
\begin{equation}
\begin{split}
\lambda_2 <p_2(1-\lambda_1-\lambda_1p_2)
\end{split}
\end{equation}
when
\begin{equation}
\begin{split}
\lambda_1 <\frac{p_1}{1+p_1p_2}
\end{split}
\end{equation}
and
\begin{equation}
\begin{split}
\lambda_1 <\frac{p_1(1-p_1-\lambda_2p_1)}{(1-p_1)}
\end{split}
\end{equation}
when
\begin{equation}
\begin{split}
\lambda_2 <\frac{p_2(1-p_1)}{1+p_1p_2}.
\end{split}
\end{equation}
for the system to be stable. 
\end{lemma}

\begin{proof}
The result in Lemma \ref{lem3} can be proved using the tool of stochastic dominance presented in \cite{Rao88}. The indistinguishability argument at the stability region boundary states that if the original system is unstable then its queues will saturate and they will always have packets to transmit; therefore at the boundaries of the stability region of the original system, the original system will be indistinguishable from the dominant system and thus has the same stability region boundaries \cite{Rao88}.
\end{proof}

The next theorem characterizes the entire stability region for the random access protocol with priorities.
\begin{theorem}\label{thm1}
The boundary of the stability region, $\mathcal{R}$, of the random access protocol with priorities, which is defined as the union of the $\mathcal{R}(\mathbf{p})$ regions for the different $\mathbf{p}=[p_1\;p_2]^T$ as
\begin{equation}
\begin{split}
\mathcal{R} =\bigcup_{\mathbf{p}\in[0,1]^2} \mathcal{R}(\mathbf{p})
\end{split}
\end{equation}
can be characterized as 
\begin{equation}
\begin{split}
\lambda_2 =\left\{\begin{array}{ll}
1-2\lambda_1& \lambda_1 \leq \frac{1}{3} \\
\frac{(1-\lambda_1)^2}{4\lambda_1} & \lambda_1 >\frac{1}{3}.
\end{array}\right. 
\end{split}
\end{equation}
\end{theorem}

\begin{proof}
%
First, we will derive the boundary of the stability region defined in lemma \ref{lem1}, which can be found as
\begin{equation}\label{P1}
\begin{split}
\lambda_2^*(\lambda_1)=&\mathbf{max}_{p_1,p_2} \;  p_2(1-\lambda_1-\lambda_1p_2) \\ &\text{\textbf{subject to }} 0\leq p_1\leq1,\;0\leq p_2\leq1,\;\lambda_1 <\frac{p_1}{1+p_1p_2}.
\end{split}
\end{equation}
Ignoring the constraints in the last optimization problem and differentiating the cost function in the last expression with respect to $p_2$ and equating the derivative to 0 we can get the optimal value for $p_2$, denoted by $p_2^*$, as\footnote{it is straightforward to prove that the cost function is concave in $p_2$.}
\begin{equation}
p_2^*=\frac{1-\lambda_1}{2\lambda_1}.
\end{equation}
Note that for $\lambda_1\geq \frac{1}{3}$, we have $p_2^*\leq 1$. Also, for $p_1=1$ and $p_2^*=\frac{1-\lambda_1}{2\lambda_1}$, the maximum value for the first queue arrival rate is $\frac{p_1}{1+p_1p_2} = \frac{2\lambda_1}{1+\lambda_1}>\lambda_1$ (i.e., the last constraint, $\lambda_1 <\frac{p_1}{1+p_1p_2}$ is satisfied with $p_1=1$), which means that for $\lambda_1\geq \frac{1}{3}$, the value for $p_2$ that maximizes $\lambda_2$ for a given $\lambda_1$ is given by $p_2^*=\frac{1-\lambda_1}{2\lambda_1}$, with all the constraints in (\ref{P1}) not being violated.

For $\lambda_1<\frac{1}{3}$, following similar steps to the $\lambda_1\geq\frac{1}{3}$ case, we can easily prove that the value for $p_2$ that maximizes $\lambda_2$ is giving by $p_2^*=1$; clearly the values of $p_1=1$ and $p_2^*=1$ can be easily checked to satisfy the constraints in (\ref{P1}) for $\lambda_1<\frac{1}{3}$.

Substituting the optimal values for $p_2$ for the different ranges of $\lambda_1$ we can easily get the boundary of the stability region spanned by the expression in lemma \ref{lem1} to be given by
\begin{equation}\label{der-bou}
\begin{split}
\lambda_2 =\left\{\begin{array}{ll}
1-2\lambda_1& \lambda_1 \leq \frac{1}{3} \\
\frac{(1-\lambda_1)^2}{4\lambda_1} & \lambda_1 >\frac{1}{3}.
\end{array}\right. 
\end{split}
\end{equation}

Finally, following a similar approach to that considered here it is straightforward to show that the boundary derived in (\ref{der-bou}) is the boundary of the stability regions defined in lemma \ref{lem2}, which completes the proof.
\end{proof}

%
%

In Fig. \ref{differentStab}, we have plotted the regions $\mathcal{R}(\mathbf{p})$, for $p_1$ and $p_2$ ranging from 0 to 1 with a step of 0.01, along with the derived stability region boundary given in the previous theorem. Fig. \ref{differentStab} also shows the stability region of the random access scheme, whose boundary is given by the following relation \cite{Rao88}
\begin{equation}
\sqrt{\lambda_1}+\sqrt{\lambda_2}=1.
\end{equation}
In Fig. \ref{differentStab}, we also show the boundary of the stability region for the time division (TD) based scheme (genie-aided), which serves as the stability region upper bound, given by\footnote{Time Division (TD) corresponds to full coordination between the two queues and requires knowledge of the queues arrival rates a priori before dividing the resources (time slots in this case).}
\begin{equation}
\lambda_1+\lambda_2=1.
\end{equation}

\begin{figure}
\centering
\includegraphics[width=.7\textwidth, height = 0.5\textwidth]{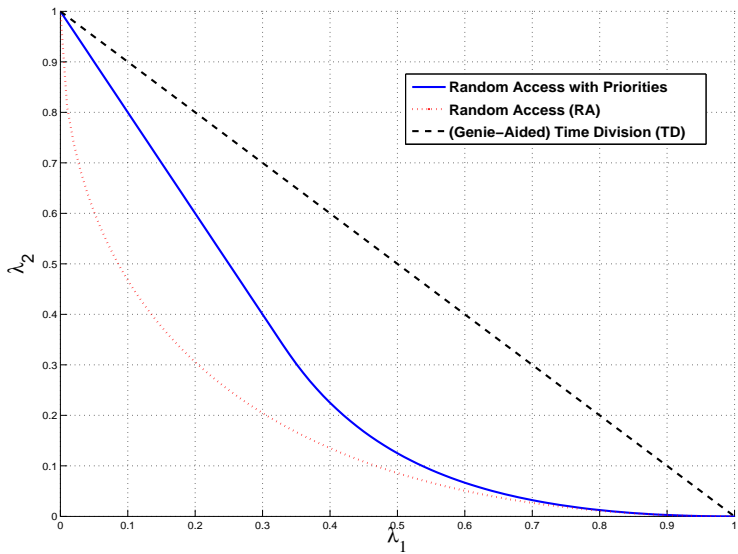}
\caption{The stability regions for the Random Access, Random Access with Priorities, and Time Division schemes.}
\label{differentStab}
\end{figure}

It is clear, and straightforward to analytically prove from the closed-form stability region boundary expressions, that the stability region for the RA scheme with priorities encloses the stability region of the RA scheme. This can explained as follows. For a given arrival rate at the first queue, $\lambda_1$, the RA with priority scheme will provide a better service rate to that queue if compared to the RA scheme and this means that queue 1 will be empty with a higher probability and this means that queue 2 will have a higher service rate as well under the RA with priority scheme as compared to the RA scheme. So setting a priority to the first queue in the retransmission will also result in a service rate improvement for the second queue; this is because the RA with priority scheme has some form of coordination between the two queues in the retransmission stage. Allowing for collision free retransmission from the first queue will decrease the amount of expected collisions between the transmissions of the two queues and this will result in better service rates for the two queues. 




\section{Conclusions}\label{Concl}
In this paper, we consider the problem of deriving the stability region for random access protocol with feedback exploitation. We consider the case of two interacting queue with priority set to one of the two queues. The two queues will access the channel through a conventional random access protocol and in the case of collision the higher priority queue will access the channel in the next slot with probability 1 while the other queue will back off. We derive the stability region for the random access with priorities protocol and prove that it contains the stability region for the conventional random access protocol. We show that not only the service rate for the higher priority queue is enhanced but also the service rate for the other queue is improved if compared to the conventional random access protocol.
\appendices

\section{Proof of Lemma \ref{lem1}}\label{app1}
In this Appendix, we provide a proof for Lemma~\ref{lem1}. We start by calculating the steady state distribution for the Markov chain shown in Fig.~\ref{fig2}.

First, it is clear that $\epsilon_0=0$ since the queue can never be in a retransmission state while being empty. Writing the balance equation around $1_R$, we have
\begin{equation}\label{1-R}
\epsilon_1=\lambda_1p_1p_2\pi_0+\left(1-\lambda_1\right)p_1p_2\pi_1.
\end{equation}
Then around $0_F$, we have
\begin{equation}\label{0-F}
(\lambda_1p_1p_2+\lambda_1(1-p_1))\pi_0=\left(1-\lambda_1\right)\epsilon_1+\left(1-\lambda_1\right)p_1(1-p_2)\pi_1.
\end{equation}
Substituting for $\epsilon_1$ from (\ref{1-R}) into (\ref{0-F}), and after some manipulations, we can get
\begin{equation}\label{00-F}
\pi_1=\frac{\lambda_1\left(1-p_1+\lambda_1p_1p_2\right)}{p_1\left(1-\lambda_1\right)\left(1-\lambda_1p_2\right)}\pi_0.
\end{equation}
Substituting from (\ref{00-F}) into (\ref{1-R}), we get
\begin{equation}\label{epis1}
\epsilon_1=\frac{\lambda_1p_2}{1-\lambda_1p_2}\pi_{0}.
\end{equation}

Writing the balance equation around $1_F$, we have
\begin{equation}\label{1-F}
\begin{split}
&\left(1-\lambda_1p_1\left(1-p_2\right)-\left(1-\lambda_1\right)(1-p_1)\right)\pi_1=\\ &\quad \lambda_1\pi_0+\lambda_1\epsilon_1+\left(1-\lambda_1\right)\epsilon_2+\left(1-\lambda_1\right)p_1\left(1-p_2\right)\pi_2.
\end{split}
\end{equation}
Around $2_R$, we have
\begin{equation}\label{2-R}
\epsilon_2=\lambda_1p_1p_2\pi_1+\left(1-\lambda_1\right)p_1p_2\pi_2.
\end{equation}
To get the relation between $\pi_1$ and $\pi_2$, we can substitute for the values of $\epsilon_1$, $\pi_0$ and $\epsilon_2$ from equations (\ref{1-R}), (\ref{0-F}) and (\ref{2-R}), respectively in equation (\ref{1-F}); after some tedious manipulation, we get
\begin{equation}\label{pi1-2}
\pi_2=\frac{\lambda_1\left(1-p_1+\lambda_1p_1p_2\right)}{p_1\left(1-\lambda_1\right)\left(1-\lambda_1p_2\right)}\pi_1.
\end{equation}
Substituting from (\ref{pi1-2}) into (\ref{2-R}), we get
\begin{equation}\label{eps2}
\epsilon_2=\frac{\lambda_1p_2}{1-\lambda_1p_2}\pi_1.
\end{equation}
Note that the Markov chain is repeating from stage 2 till the end. For $k\geq 2$, we have the following relations. 
\begin{equation}\label{pik}
\pi_k=\frac{\lambda_1\left(1-p_1+\lambda_1p_1p_2\right)}{p_1\left(1-\lambda_1\right)\left(1-\lambda_1p_2\right)}\pi_{k-1}.
\end{equation}
\begin{equation}\label{episk}
\epsilon_k=\frac{\lambda_1p_2}{1-\lambda_1p_2}\pi_{k-1}.
\end{equation}
The last relation can be used to prove the following relation between $\epsilon_k$ and $\epsilon_{k-1}$.
\begin{equation}\label{pik}
\epsilon_k=\frac{\lambda_1\left(1-p_1+\lambda_1p_1p_2\right)}{p_1\left(1-\lambda_1\right)\left(1-\lambda_1p_2\right)}\epsilon_{k-1}.
\end{equation}

The steady state distribution can now be written as follows.
\begin{itemize}
\item $\epsilon_0=0$.
\item $\pi_k=\rho^{k}\pi_0$, $k\geq1$ and $\rho=\frac{\lambda_1\left(1-p_1+\lambda_1p_1p_2\right)}{p_1\left(1-\lambda_1\right)\left(1-\lambda_1p_2\right)}$.
\item $\epsilon_1=\frac{\lambda_1p_2}{1-\lambda_1p_2}\pi_{0}$.
\item $\epsilon_k=\rho^{k-1}\epsilon_1$, $k\geq2$.
\end{itemize}
This steady state distribution can be easily checked to satisfy the balance equation at any general state (details are omitted since it is a rather straightforward, yet very tedious, procedure).

To get the value of the steady state probabilities, we apply the following normalization requirement.
\begin{equation}\label{norm1}
\begin{split}
&\sum_{k=0}^{\infty}(\pi_k+\epsilon_{k})=1\\&\quad \quad \rightarrow \pi_0+\sum_{k=1}^{\infty}(\pi_k+\epsilon_{k})=\pi_0\left(1+\frac{\lambda_1p_2}{1-\lambda_1p_2}\right)\sum_{k=0}^{\infty}\rho^{k}=1,
\end{split}
\end{equation}
where $\rho=\frac{\lambda_1\left(1-p_1+\lambda_1p_1p_2\right)}{p_1\left(1-\lambda_1\right)\left(1-\lambda_1p_2\right)}$ as defined above. 

Note that for the steady state distribution to exist, i.e. to have $\pi_0$ to be non zero, then we must have $\rho<1$, which is the stability condition for queue 1 in this dominant system. Therefore, the stability condition can be stated as
 \begin{equation}\label{stab-dom1}
\rho<1\;\rightarrow \lambda_1<\frac{p_1}{1+p_1p_2}.
\end{equation}

From the normalization condition in (\ref{norm1}), we can get the value of $\pi_0$ as
\begin{equation}\label{pi-0}
\pi_0=\frac{p_1-\lambda_1(1+p_1p_2)}{p_1(1-\lambda_1)}.
\end{equation}

In Dominant System 1, queue 2 will be served only in the states denoted by the subscript $F$ in Fig.~\ref{fig2} since in the retransmission states, denoted by the subscript $R$  in Fig.~\ref{fig2}, queue 2 will be in the back off mode. Hence, the service rate, $\mu_2$, for queue 2 in Dominant System 1 is given by
\begin{equation}\label{Dom2Qu2Ser}
\begin{split}
\mu_2&=p_2(1-\lambda_1)\pi_0+p_2(1-p_1)\lambda_1\pi_0+\sum_{k=1}^{\infty}p_2(1-p_1)\pi_k\\&=p_2(1-p_1\lambda_1)\pi_0+\sum_{k=1}^{\infty}p_2(1-p_1)\pi_k,
\end{split}
\end{equation}
where in the $0_F$ state, and with the arrival at the beginning of the slot assumption, queue 2 is served with a rate of $p_2(1-\lambda_1)\pi_0$ with no arrival at the beginning of the slot since queue 1 will not attempt any random access since it is empty, and $p_2(1-p_1)\lambda_1\pi_0$ with arrival at the slot beginning; for the other first transmission states, queue 2 will be served if it decides to access the medium, which occurs with probability $p_2$, and queue 1 decides not to access the medium, which occurs with probability $(1-p_1)$. After some manipulation, we can write the expression for $\mu_2$ as
\begin{equation}\label{Dom2Qu2SerFin}
\mu_2=p_2(1-\lambda_1-\lambda_1p_2).
\end{equation}
For the stability of queue 2, we must have
\begin{equation}\label{Dom2Qu2Arr}
\lambda_2<\mu_2=p_2(1-\lambda_1-\lambda_1p_2).
\end{equation}

\section{Proof of Lemma \ref{lem2}}\label{app2}
In this Appendix, we provide a proof for Lemma~\ref{lem2}. We start by calculating the steady state distribution for the Markov chain shown in Fig.~\ref{fig3}. The state transition matrix, $\mathbf{\Phi}$, of the Markov chain shown in Fig.~\ref{fig3} can be written as
\begin{equation}
\mathbf{\Phi}=\left(   \begin{array}{ ccccc}
\mathbf{B}& \mathbf{A}_0& \mathbf{0} & \mathbf{0}  &\cdots \\ \mathbf{A}_2& \mathbf{A}_1& \mathbf{A}_0 & \mathbf{0}  &\cdots
\\  \mathbf{0}  & \mathbf{A}_2& \mathbf{A}_1& \mathbf{A}_0 & \cdots \\  \mathbf{0} &\mathbf{0}  & \mathbf{A}_2& \mathbf{A}_1 & \cdots \\  \vdots &\vdots  & \vdots & \vdots & \ddots 
\end{array}\right)
\end{equation}
where 
\begin{equation*}
\begin{split}
\mathbf{B}&=\left(   \begin{array}{ cc}
(1-\lambda_2)+\lambda_2(1-p_1)p_2 & 0 \\ 0&0
\end{array}\right),
\\
\mathbf{A}_0&=\left(   \begin{array}{ cc}
(1-\lambda_2)(1-p_1)p_2 & 0 \\ 0&0
\end{array}\right),
\\
\mathbf{A}_1&=\left(   \begin{array}{ cc}
\lambda_2p_2(1-p_1)+(1-\lambda_2)(1-p_2) & 1-\lambda_2 \\ (1-\lambda_2 )p_1p_2&0
\end{array}\right),
\\
\mathbf{A}_2&=\left(   \begin{array}{ cc}
\lambda_2 & \lambda_2 \\ \lambda_2p_1p_2 & \lambda_2
\end{array}\right).
\end{split}
\end{equation*}
The steady state distribution vector is given by $\mathbf{v}=[\pi_0'\;\epsilon_0'\;\pi_1'\;\epsilon_1'\;\pi_2'\;\epsilon_2'\;\cdots]^T$ and $\mathbf{v} =\mathbf{\Phi} \mathbf{v}$.

The state transition matrix $\mathbf{\Phi}$ is a block-tridiagonal matrix; therefore the Markov chain shown in Fig.~\ref{fig3} is a homogeneous quasi birth-and-death (QBD) Markov chain \cite{latouche}. Note that to make the state transition matrix a block-tridiagonal matrix we have added a transition from the $0_{\text{OFF}}$ state to the $1_{\text{ON}}$ state as shown in Fig. \ref{FB2add} and this will preserve the structure of the state transitions between the different stages in the Markov chain. Note that adding this transition will not affect the stationary state distribution of the Markov chain as well as the balance equations since $\epsilon_0'=0$ even with the added transition since the Markov chain will never enter the $0_{\text{OFF}}$ state\footnote{The analysis presented here could have been used for analyzing the Markov chain shown in Fig.~\ref{fig2}; however, the structure of this Markov chain allowed for the use of a simpler approach that was adopted in Appendix \ref{app1}}. 

\begin{figure}
\centering
\includegraphics[width=.7\textwidth, height = 0.5\textwidth]{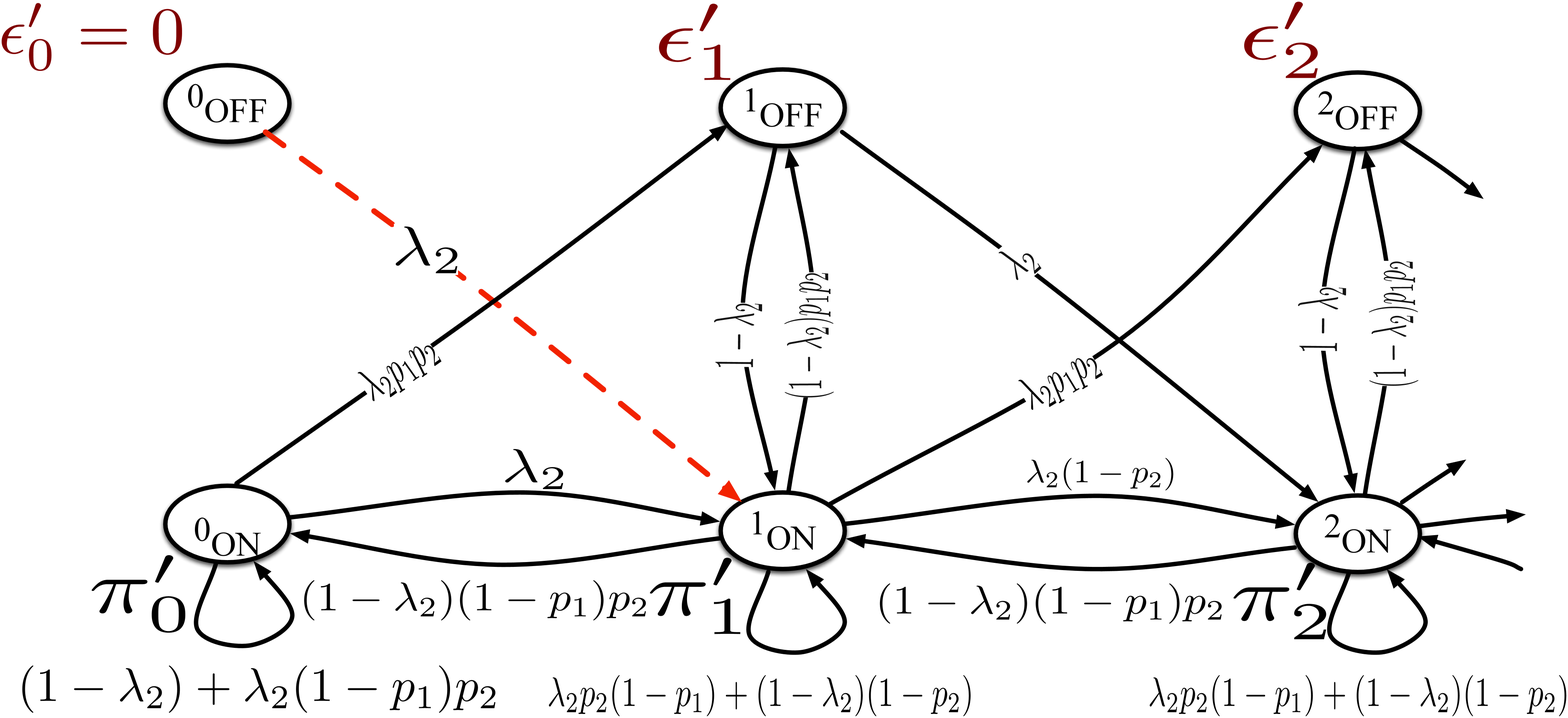}
\caption{The queue 2 Markov chain with added transition between $0_{\text{OFF}}$ and $1_{\text{ON}}$ to make the state transition matrix a block-tridiagonal matrix. }
\label{FB2add}
\end{figure}

Define the vector $\mathbf{v}_k'=[\pi_k' \;\epsilon_k']^T$. Note that $\mathbf{v}_0'=[\pi_0' \;0]^T$. The steady state distribution of the Markov chain shown in Fig. \ref{fig3} satisfies the following equation \cite{latouche}
\begin{equation}
\mathbf{v}'_k=\mathbf{R}^k\mathbf{v}_0',\quad k\geq 1,
\end{equation}
where the $2\times 2$ matrix $\mathbf{R}$ is given by the solution to the following equation.
\begin{equation}\label{R-balance}
\mathbf{A}_0+\mathbf{R}(\mathbf{A_1}-\mathbf{I}_{2})+\mathbf{R}^2\mathbf{A_2}=\mathbf{0}_{2\times 2},
\end{equation}
where $\mathbf{I}_2$ is the $2\times 2$ identity matrix and $\mathbf{0}_{2\times 2}$ is the all zeros $2 \times 2$ matrix.

To get the stationary distribution, we have to find the matrix 
\begin{equation*}
\mathbf{R}=\left(   \begin{array}{ cc}
r_{11}&r_{12}\\ r_{21}&r_{22}
\end{array}\right).
\end{equation*}
Note that for $\mathbf{v}_1'=\mathbf{R}\mathbf{v}_0'$, where $\mathbf{v}_0'=[\pi_0'\; 0]^T$ and $\mathbf{v}_1'=[\pi_1'\; \epsilon_1']^T$. Therefore, we have 
\begin{equation}
\begin{split}
r_{11}=\frac{\pi_1'}{\pi_0'} \text{ and }r_{21}=\frac{\epsilon_1'}{\pi_0'}.
\end{split}
\end{equation}
Writing the balance equation around the $0_{\text{ON}}$ in Fig.~\ref{fig3}, we have
\begin{equation}\label{Dom2-1}
\begin{split}
&(\lambda_2p_1p_2+\lambda_2(1-p_2))\pi_0'=(1-\lambda_2)(1-p_1)p_2\pi_1'\\&\quad\rightarrow \pi_1' =\frac{\lambda_2(1-p_2+p_1p_2)}{(1-\lambda_2)(1-p_1)p_2}\pi_0'.
\end{split}
\end{equation} 
Therefore, we have
\begin{equation}\label{Dom2-r11}
\begin{split}
r_{11} =\frac{\lambda_2(1-p_2+p_1p_2)}{(1-\lambda_2)(1-p_1)p_2}.
\end{split}
\end{equation} 

Writing the balance equation around $1_{\text{OFF}}$, we have 
\begin{equation}\label{Dom2-2}
\epsilon_1'=\lambda_2 p_1p_2\pi_0'+(1-\lambda_2) p_1p_2\pi_1'\rightarrow \epsilon_1' =\frac{\lambda_2 p_1}{1-p_1}\pi_0'.
\end{equation} 
Therefore, we have
\begin{equation}\label{Dom2-r21}
r_{21} =\frac{\lambda_2 p_1}{1-p_1}.
\end{equation}

To get the values of $r_{12}$ and $r_{22}$, we consider the transition across the border shown in Fig. \ref{border}. For the Markov chain to be positive recurrent then the probability of going across the border in both directions must be the same \cite{gallager-book}; hence, we have
\begin{equation}\label{Dom2-rr1}
\begin{split}
(1-\lambda_2)(1-p_1)p_2\pi_2'=(\lambda_2p_1p_2+\lambda_2(1-p_2))\pi_1'+\lambda_2\epsilon_1'.
\end{split}
\end{equation} 
But we have $\mathbf{v}_2'=\mathbf{R}\mathbf{v}_1'$, from which we have $\pi_2'=r_{11}\pi_1'+r_{12}\epsilon_1'$; using (\ref{Dom2-r11}) and (\ref{Dom2-rr1}), we can easily show that
\begin{equation}\label{Dom2-r12}
\begin{split}
r_{12} =\frac{\lambda_2}{(1-\lambda_2)(1-p_1)p_2}.
\end{split}
\end{equation} 

\begin{figure}
\centering
\includegraphics[width=.7\textwidth, height = 0.5\textwidth]{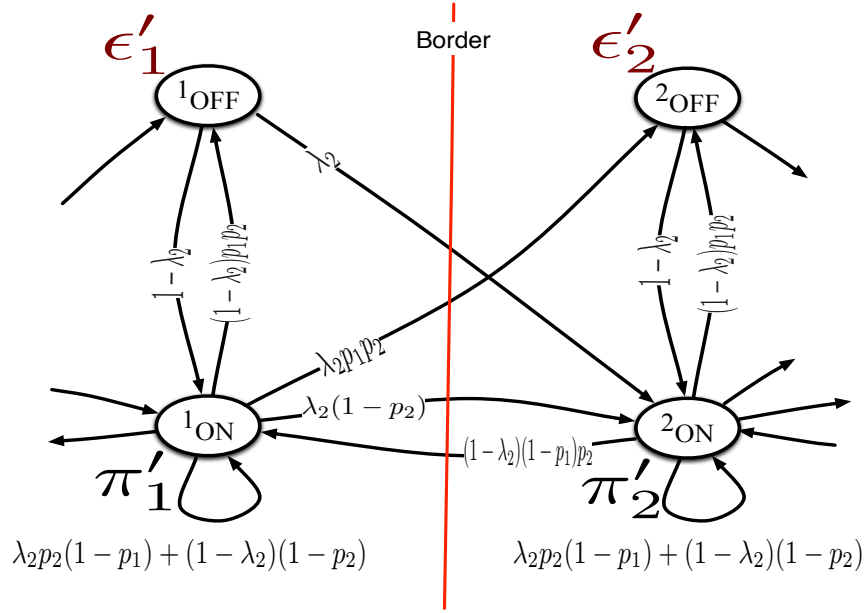}
\caption{The segment of the Markov chain used to calculate the values of $r_{12}$ and $r_{22}$.}
\label{border}
\end{figure}

Finally, the balance equation around $2_{\text{OFF}}$ can be written as \setcounter{equation}{40}
\begin{equation}\label{Dom2-5}
\epsilon_2'= \lambda_2p_1p_2\pi_1'+(1-\lambda_2)p_1p_2\pi_2' =r_{21}\pi_1'+r_{22}\epsilon_2'.
\end{equation} 
Substituting for $\pi_2'$ from (\ref{Dom2-rr1}), we can easily show that
\begin{equation}\label{Dom2-r22}
r_{22} =\frac{\lambda_2 p_1}{1-p_1}.
\end{equation} 

Now the matrix $\mathbf{R}$ is given by
\begin{equation}\label{MatrixR}
\mathbf{R}=\left(   \begin{array}{cc}
\frac{\lambda_2(1-p_2+p_1p_2)}{(1-\lambda_2)(1-p_1)p_2} & \frac{\lambda_2}{(1-\lambda_2)(1-p_1)p_2} \\ \frac{\lambda_2 p_1}{1-p_1} & \frac{\lambda_2 p_1}{1-p_1}
\end{array}\right),
\end{equation}
which can be easily checked to satisfy the balance equation given by (\ref{R-balance}).

To get the stationary distribution of the Markov chain shown in Fig.~\ref{fig3}, we apply the following normalization requirement.
\begin{equation}\label{norm22}
\sum_{k=0}^{\infty}(\pi_k'+\epsilon_{k}')=1\rightarrow [1\;1]\left(\sum_{k=0}^{\infty}\mathbf{R}^k\right)\mathbf{v}_0'=1.
\end{equation}
For the summation $\left(\sum_{k=0}^{\infty}\mathbf{R}^k\right)$ to converge we must have the spectral radius of the matrix $\mathbf{R}$, $\textbf{sp}(\mathbf{R})$, to be less than one \cite{latouche}\footnote{The spectral radius of a matrix is the maximum over the magnitudes of its eigenvalues.}. From (\ref{MatrixR}), we can easily get $\textbf{sp}(\mathbf{R})$ to be given by 
\begin{equation}\label{Dom2-SP}
\begin{split}
\textbf{sp}(\mathbf{R}) = \frac{\lambda_2 \left(1- p_2 - \lambda_2p_1p_2 + 2p_1p_2+ \sqrt{1 - 2p_2 + p_2^2 + 4p_1p_2   - 2\lambda_2p_1p_2 - 2\lambda_2 p_1 p_2^2  +\lambda_2^2p_1^2p_2^2}\right)}{2p_2\left(1 - \lambda_2 - p_1 + \lambda_2p_1\right)}.
\end{split}
\end{equation}

The requirement that $\textbf{sp}(\mathbf{R})<1$ can be used in the last expression to get the stability condition of the second queue arrival rate $\lambda_2$ as
\begin{equation}
\begin{split}
\lambda_2<\frac{p_2(1-p_1)}{1+p_1p_2}.
\end{split}
\end{equation} 

Going back to the normalization requirement in (\ref{norm22}), we have
\begin{equation}\label{no22}
[1\;1]\left(\sum_{k=0}^{\infty}\mathbf{R}^k\right)\mathbf{v}_0'=[1\;1]\left(\mathbf{I}_2-\mathbf{R}\right)^{-1} \left[\begin{array} {c}\pi_0' \\ 0\end{array}\right]=1.
\end{equation}
From the last expression, we can easily prove that $\pi_0'$ is given by
\begin{equation}\label{pi0}
\pi_0' = \frac{p_2-\lambda_2 - p_1p_2 - \lambda_2 p_1p_2}{(1-\lambda_2)(1-p_1)p_2}.
\end{equation}
Note that the requirement that $\pi_0'>0$, i.e. a non-zero empty queue probability, is satisfied if $\lambda_2<\frac{p_2(1-p_1)}{1+p_1p_2}$, which is the queue stability condition.

The service rate, $\mu_1'$, for queue 1 in Dominant System 2 can now be expressed as
\begin{equation}\label{mu22}
\begin{split}
\mu_1' &= p_1(1-\lambda_2)\pi_0'+p_1(1-p_2)\lambda_2\pi_0'+p_1(1-p_2)\sum_{k=2}^{\infty}\pi_k'+\sum_{k=2}^{\infty} \epsilon_{k}'
\\&=p_1(1-\lambda_2p_2)\pi_0'+[p_1(1-p_2)\;1]\left(\sum_{k=1}^{\infty}\mathbf{R}^k\right) \left[\begin{array} {c}\pi_0' \\ 0\end{array}\right]
\\&=p_1(1-\lambda_2p_2)\pi_0'+[p_1(1-p_2)\;1]\;\mathbf{R}\left(\mathbf{I}_2-\mathbf{R}\right)^{-1} \left[\begin{array} {c}\pi_0' \\ 0\end{array}\right]
\\&= \frac{p_1(1-p_1-\lambda_2p_1)}{(1-p_1)},
\end{split}
\end{equation}
where in the OFF states, queue 1 is served with probability 1 since queue 2 will be in the back off mode. For the stability of queue 1 in Dominant System 2 we must have
\begin{equation}\label{stab22}
\lambda_1<\mu_1' =\frac{p_1(1-p_1-\lambda_2p_1)}{(1-p_1)},
\end{equation}
which completes the proof.


\bibliographystyle{IEEEbib}
\bibliography{MyLib}

\end{document}